\documentclass{amsart}

\usepackage[]{graphicx}
\usepackage[text={7in,10in},centering]{geometry}
\usepackage{color}
\usepackage{amsmath}
\usepackage{amsfonts}
\usepackage{graphicx}
\usepackage{float}

\newtheorem{theorem}{Theorem}[section]

\theoremstyle{definition}
\newtheorem{definition}[theorem]{Definition}

\theoremstyle{remark}

\numberwithin{equation}{section}

\definecolor{lightgray}{gray}{0.5}
\begin{document}

\title{A Review of the Unknown Input Observer with Examples}

\author{Sam Nazari}
\address{Department of Electrical \& Computer Engineering, Northeastern University, Boston, MA,02210}
\email{nazari.s@neu.edu}





\begin{abstract}
We consider the Unknown Input Observer and focus on three examples found in the literature.
\end{abstract}

\maketitle
\section{INTRODUCTION}
\label{sec:intro}  

This memo considers the design of observers for a class of linear dynamic systems in which system uncertainty can be modeled as an additive unknown disturbance term in the dynamic equation: 
\begin{align} \label{eq:linSys}
\dot{x}(t) &= Ax(t) + Bu(t) + Ed(t) \\
y(t)       &= Cx(t) \nonumber
\end{align}

\noindent where $x(t) \in \mathbb{R}^n$ is the state vector, $y(t) \in \mathbb{R}^m$ is the output vector, $u(t) \in \mathbb{R}^r$ is the known input vector and $d(t) \in \mathbb{R}^q$ is the unknown input or disturbance vector. If the unknown input distribution matrix, $E$, is not full column rank, then the rank decomposition: 

\begin{align*}
Ed(t) = E_1 E_2 d(t)
\end{align*}

can be applied where, $E_1$, is a full column rank matrix and $E_2 d(t)$ can be considered as a new unknown input. The term, $Ed(t)$, may be used to model additive disturbances or model uncertainties that are not known \textit{a priori}. Typical uncertainties are noise, plant nonlinearities, model reduction errors, parametric uncertainties and difficult to model interconnection terms in large scale systems. While not considered in this memo, a disturbance term may also appear in the output equation:
\begin{align*}
y(t) = Cx(t)+E_y d(t)
\end{align*}
Additionally, this memo does not consider the $Du(t)$ term in the system output equations.  For the sake of brevity and without loss of generality, we will not treat systems with this term. Note that the development of the theoretical material in this memo closely follows Chapter 3 of \cite{Chen99}. 
\section{Unknown Input Observer Theory} 

The problem of designing observers for a linear system with both known and unknown inputs has been studied since the 1970's \cite{Chen99}, \cite{Pat89}, \cite{Shaf2015}.  The motivation for studying this problem is that in practice, for a variety of reasons, many plants are modeled with disturbance terms as shown in Eq \ref{eq:linSys}. When modeled in such a way, traditional observers with the Luenberger structure become difficult to implement since they use all of the input signals in order to compute an estimate of the state vector. This limitation renders them ineffective for many practical applications. Contrastingly, the observer structure considered in this memo assumes no \textit{a priori} knowledge about the unknown inputs, rendering it more amenable for use in fault detection applications. Many investigators have addressed this problem with varying success and we refer the reader to Section 3.2 of \cite{Chen99} for a broad literature review. The class of observers examined in this memo is based on the \textit{Unknown Input Observer} (UIO) scheme proposed by Chen, Patton and Zhang.  

\begin{definition}
An observer is defined as an \textit{unknown input observer} \cite{Sch12,Pat89} for the system described by Eq \ref{eq:linSys}, if its state estimation error vector, $e(t)$, defined as:

\begin{align}
e(t) = x(t) - \hat{x}(t)
\end{align}

approaches zero asymptotically regardless of the presence of unknown inputs, $d(t)$, in the system.  Furthermore, the structure for a full order UIO is given by the dynamic system: 

\begin{equation} \label{eq:uio}
\begin{aligned}
\dot{z}(t) &= Fz(t)+TBu(t) + Ky(t) \\
\hat{x}(t) &= z(t)+Hy(t)
\end{aligned}
\end{equation}
\end{definition}

\begin{figure}[H]
    \centering
    \includegraphics[scale=0.4]{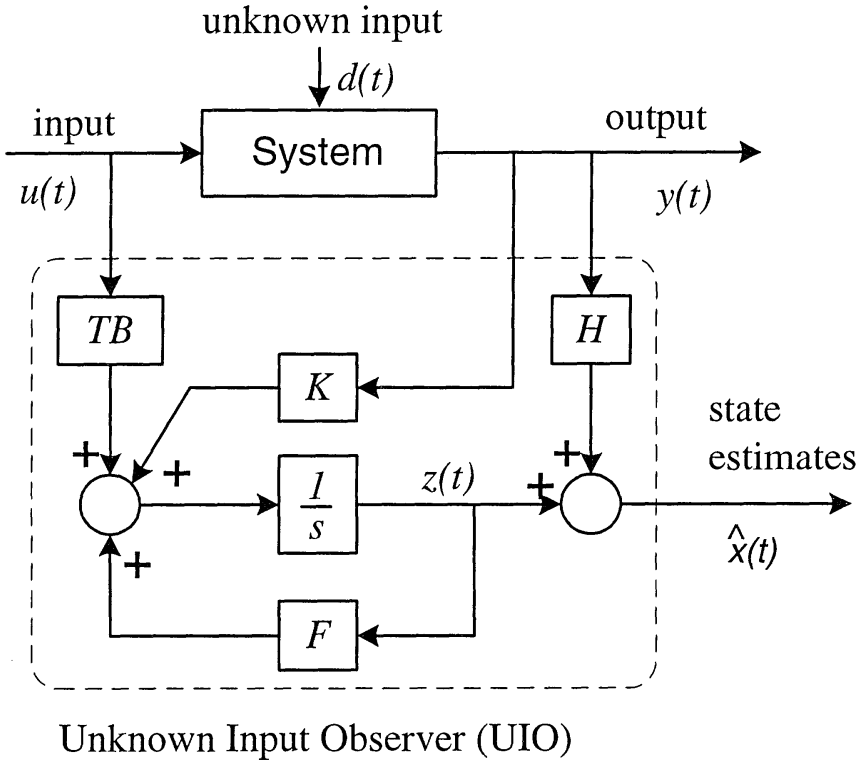}
    \caption{UIO Structure where $K = K_1 + K_2$.}
    \label{fig:UIOBD}
\end{figure}

where $\hat{x} \in \mathbb{R}^n$ is the state estimate, $z \in \mathbb{R}^n$ is the state of the full-order dynamic observer, and $F,T,K,H$ are matrices to be designed for achieving unknown input decoupling. A block diagram of the UIO is shown in Figure \ref{fig:UIOBD}.  From the block diagram, one can see that the UIO is essentially a dynamic system that decouples the state estimation dynamics from the disturbance term in the original system. By expanding $\dot{e}(t)$ one obtains: 

\begin{align*}
\dot{e}(t) &= (A-HCA-K_1 C)e(t) + \left [ F-(A-HCA-K_1 C) \right ]z(t) \\
&+[K_2 - (A-HCA-K_1 C)]y(t) \\
&+ [T-(I-HC)]Bu(t) + (HC-I)Ed(t)
\end{align*}

and it is easy to see that in order to make the estimation error a function of $Fe(t)$:

\begin{align}
\dot{e}(t) = Fe(t) \label{eq:efe}
\end{align}

the equations: 

\begin{align}
0       &= (HC-I)E \label{eq:HCIE}\\
T       &= I - HC \\ 
F       &= A-HCA-K_1 C \label{eq:F}\\
K_2     &= FH
\end{align}
must hold. If all the eigenvalues in $F$ are stable, then $e(t)$ will approach zero asymptotically.  Note that Eq \ref{eq:efe} is not a function of $E$ or $d$.  Therefore, the estimation error approaches zero independently of the disturbance terms, achieving the desired decoupling of the state estimate from the unknown disturbance inputs.  

\begin{theorem}The necessary and sufficient conditions for the system in Eq \ref{eq:uio} to be an UIO for the system defined in \ref{eq:linSys} are: 
\begin{enumerate}
    \item $rank(CE)=rank(E)$
    \item $(C,A_1)$ is a detectable pair, where $A_1 = A-E[(CE)^T CE]^{-1} (CE)^T CA$
\end{enumerate}
\end{theorem}

\begin{proof}
Please see \cite{Chen99}
\end{proof}
According to \cite{Shaf2015}, the disturbance estimate can now be obtained, if needed, by the relation:

\begin{equation}
\begin{aligned}
\hat{d} = (CE)^{\dagger}[\dot{\hat{y}}-CA\hat{x}-CBu]
\end{aligned}
\end{equation}

Next we provide a step by step design procedure for the UIO that will be utilized in the examples that ensue. 
\subsection{Design Procedure for UIOs} 
\label{sec: UIODesign}
 
We will now describe a procedure that can lead to effective UIO design.  It can be shown that Eqn \ref{eq:HCIE} is solvable if and only if the condition $rank(CE) = rank(E)$ is met \cite{Chen99}. Therefore, the rank condition  $rank(CE) = rank(E)$ must first be checked. If this condition is not met, then a UIO does not exist.  If the condition is met, then we may compute the following observer matrices: 

\begin{center}
\begin{align*}
H   &= E[(CE)^T CE]^{-1} (CE)^T\\ 
T   &= I-HC \\
A_1 &= TA
\end{align*}
\end{center}
Next, the observability of the pair $(C,A_1)$ must be checked. As shown in Eq \ref{eq:efe}, an important step in designing a UIO is to stabilize $F=A_1-K_1 C$.  If the pair $(C,A_1)$ is detectable, then this can be achieved by using pole placement to choose the appropriate gain matrix $K_1$.  It can be shown that the observability of the pair $(C,A_1)$ is equivalent to the observability of the pair $(C,A)$ \cite{Sch12}. Therefore, if the pair $(C,A_1)$ is observable or at least detectable then the gain, $K_1$, in Eq \ref{eq:F} can be obtained for the UIO. Hence, if $(C,A_1)$ is observable then a UIO exists and $K_1$ should be computed using pole placement. 
\vspace{3mm}

On the other hand, if the pair $(C, A_1)$ is not observable, then one must construct a transformation matrix, $P$, for the observable canonical decomposition of the system:

\begin{align*}
PA_1P^{-1} &= 
\left [ \begin{array}{cc} 
A_{11} & 0 \\ 
A_{12} & A_{22} 
\end{array} \right ] \hspace{5mm} A_{11} \in \mathbb{R}^{n_1 \times n_1} \\
CP^{-1} &= \left [ C^* \hspace{2mm} 0 \right ] \hspace{17mm} C^* \in \mathbb{R}^{m \times n_1}
\end{align*}

where $n_1$ is the rank of the observability pair $(C,A_1)$, and $(C^*, A_{11})$ is observable. If any of the eigenvalues of $A_{22}$ are unstable, then a UIO does not exist. If the eigenvalues of $A_{22}$ are stable, then $(C, A_1)$ is detectable and we may select $n_1$ desirable eigenvalues and assign them to $A_{11} - K_p^1 C^*$ using pole placement to obtain $K_p^1$. Next compute 

\begin{align*}
K_1 = P^{-1}K_p = P^{-1}[(K_p^1)^T \hspace{2mm} (K_p^2)^T]^T
\end{align*} 

where $K_p^2$ can be any $(n-n1) \times m$ matrix, because it does not affect the eigenvalues of $F$. Finally, compute the observer matrices:
\begin{align*}
F &= A_1 - K_1 C \\ 
K &= K_1 + K_2 = K_1 + FH.
\end{align*}

In order to clarify the design procedure we consider some numerical examples next.
\section{Numerical Examples} 
\subsection{Third order system with no inputs}\label{ex:ex1}
The first example considered is taken from page 76 of \cite{Chen99}. There, the following parameter matrices are used: 

\begin{equation}
\begin{aligned}
A = \displaystyle \left [ \begin{array}{rrr}
-1 & 1 & 0 \\
-1 & 0 & 0 \\
 0 & -1 & -1
 \end{array} \right ], \hspace{2mm} C = \displaystyle \left [ \begin{array}{rrr} 1 & 0 & 0 \\ 0 & 0 & 1 \end{array} \right ], \hspace{2mm} E = \displaystyle \left [ \begin{array}{rrr} -1 & 0 & 0 \end{array} \right ]
\end{aligned}
\end{equation}

\begin{par}
\noindent First we check that $rank(CE) = rank(E) = 1$:
\end{par} 
\begin{verbatim}
rank(C*E)
rank(E)
\end{verbatim}

        \color{lightgray} \begin{verbatim}
ans =       ans =

      1             1
\end{verbatim} \color{black}

\begin{par}
The rank is equal to $1$ as needed. The observer matrices are computed next:
\end{par}

\begin{verbatim}
H = E*inv((C*E)'*(C*E))*(C*E)'
T = eye(3)-H*C
A1 = T*A
\end{verbatim}
        \color{lightgray} \begin{verbatim}
H =                 T =                 A1 =

     1     0           0    0   0            0   0   0
     0     0           0    1   0           -1   0   0
     0     0           0    0   1            0  -1  -1

\end{verbatim} \color{black}

\begin{par}
Next, the observability of the pair $(A_1,C)$ must be checked. If it is full rank ($rank(obsv(A_1,C))=3$), then we may apply pole placement:
\end{par}
\begin{verbatim}
rank(obsv(A1,C))
\end{verbatim}

        \color{lightgray} \begin{verbatim}
ans =

     3

\end{verbatim} \color{black}
    
Now we proceed to apply pole placement in order to place the observer poles at $[-2,-10,-5]$.

\begin{verbatim}
K1 = place(A',C',[-2,-10,-5])'
\end{verbatim}

        \color{lightgray} \begin{verbatim}
K1 =

    6.7231   -3.7816
    7.5581  -11.3458
   -3.0543    8.2769

\end{verbatim} \color{black}

\begin{par}
Finally, the observer $F$ and $K$ matrices are computed:
\end{par} 
\begin{verbatim}
F = A1-K1*C
K = K1 + F*H
\end{verbatim}
    \color{lightgray} \begin{verbatim}
F =                                  K =

   -6.7231         0    3.7816              0    -3.7816
   -8.5581         0   11.3458             -1.0  -11.3458
    3.0543   -1.0000   -9.2769              0     8.2769

\end{verbatim} \color{black}

\begin{figure}[h]
    \centering
    \includegraphics[scale=0.45]{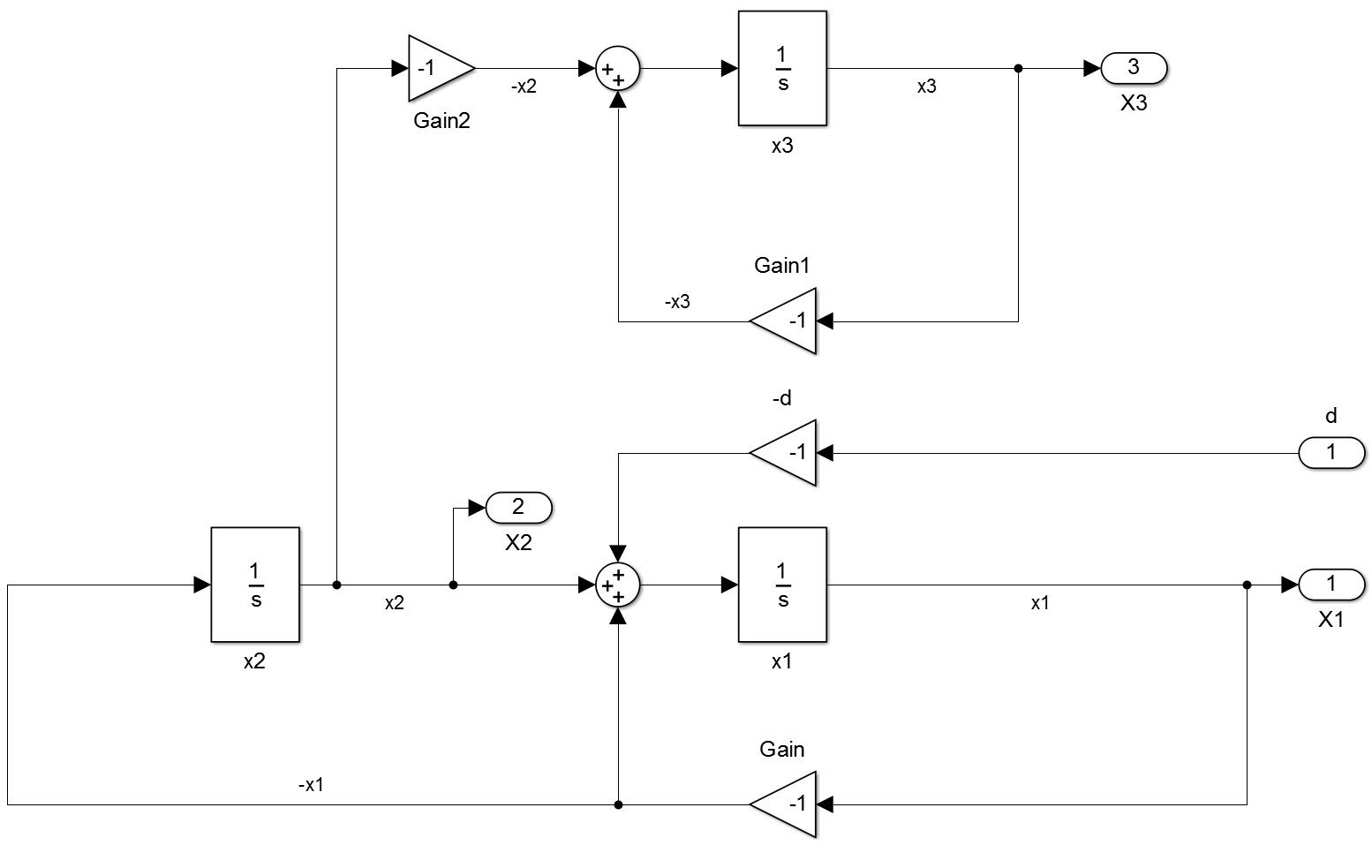}
    \caption{Simulink block diagram of the dynamics subsystem.}
    \label{fig:dynamics}
\end{figure}
\begin{figure}[h]
    \centering
    \includegraphics[scale=.45]{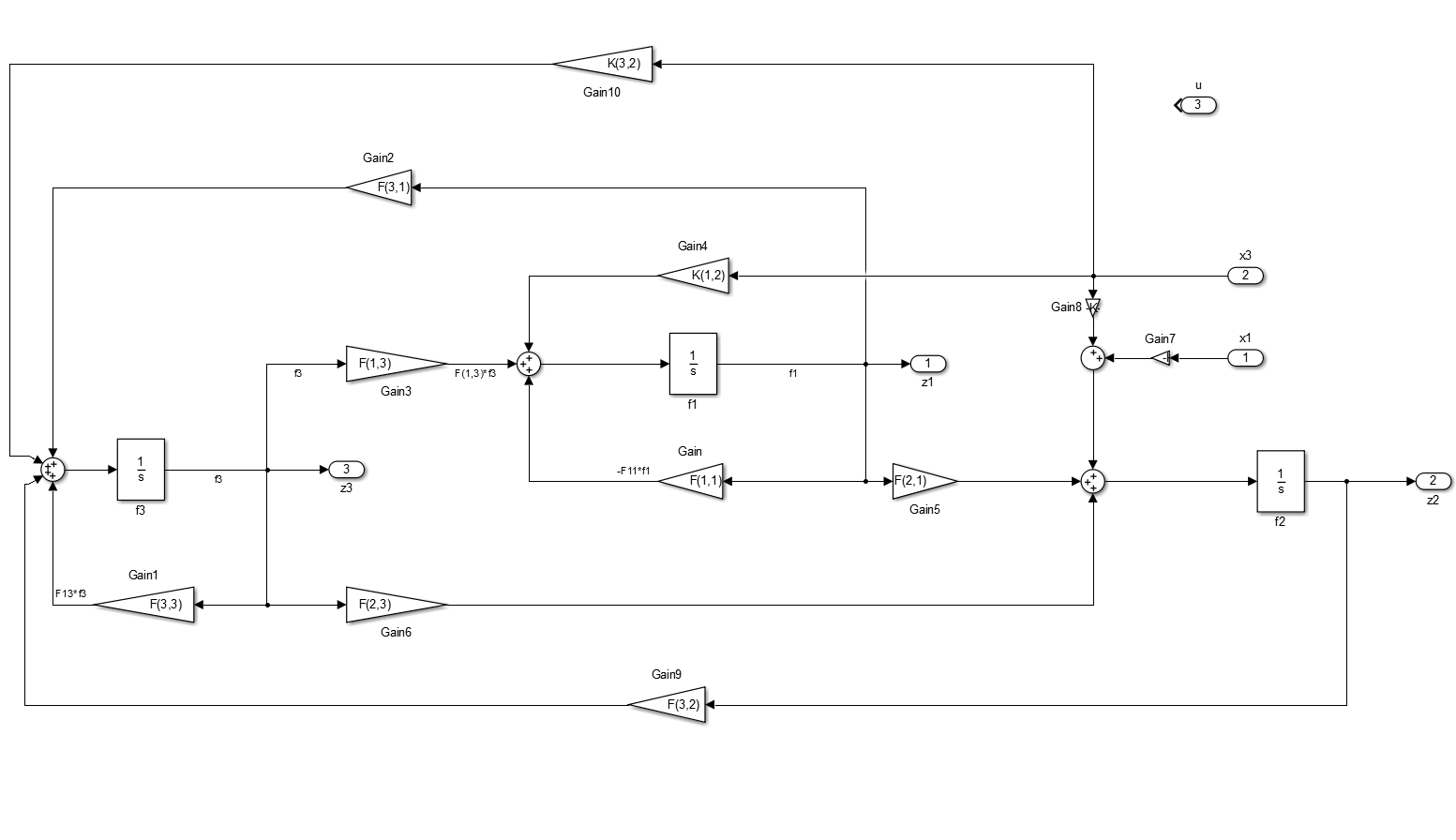}
    \caption{Simulink block diagram of the UIO.}
    \label{fig:UIO}
\end{figure}
Let us examine the Simulation results with a disturbance occurring at $t=10$ seconds. This is showin in Figure \ref{fig: resultsofUIO}. The initial conditions of the dynamic system as:
\begin{verbatim}
x1_0 = 100;
x2_0 = -100;
x3_0 = 1;
\end{verbatim}

In Figure \ref{fig: resultsofUIO}, one is able to see the state trajectories in the solid lines.  The dotted 'x' represent the estimate from the UIO.  For the given disturbance and the chosen observer pole locations, the observer seems to converge (remove all estimation error) within five seconds.  Figures \ref{fig:dynamics} and \ref{fig:UIO} show the block diagram in detail.  Note that we have chosen not to estimate the disturbance itself. All the files for this simulation can be found in the \color{lightgray}\begin{verbatim}C:\Users\sqn4594\Documents\FDIR \end{verbatim}\color{black} folder. Specifically, the file \begin{verbatim} bookEx.m \end{verbatim} and the Simulink model \begin{verbatim} exOne.xsl \end{verbatim} contain everything. 
\begin{figure} [H]
    \centering
    \includegraphics [width=3.5in]{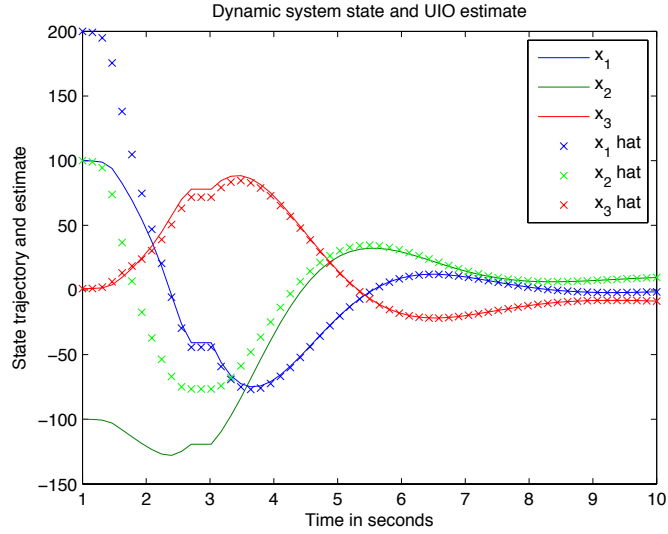}
    \caption{State response and UIO estimate for example}
    \label{fig: resultsofUIO}
\end{figure}

\subsection{Car Suspension System}\label{ex:ex2}
The next example is a quarter car model.  The system matrices are: 

\begin{equation}
\begin{aligned}
A = \displaystyle \left [ \begin{array}{rrrr}
0 & 1 & 0 & 0 \\
\frac{-ks}{mb} & \frac{-bs}{mb} & \frac{ks}{mb} & \frac{bs}{mb} \\
 0 & 0 & 0 & 1 \\
 \frac{ks}{mw} & \frac{bs}{mw} & \frac{-ks-kt}{mw} & \frac{-bs}{mw}
 \end{array} \right ] \hspace{2mm} B= \displaystyle \left [ \begin{array}{rr} 
 0 & 0\\
 0 & \frac{1000}{mb}\\
 0 & 0 \\
 \frac{kt}{mw} & \frac{-1000}{mw}
 \end{array} \right ] \hspace{2mm} C = \displaystyle \left [ \begin{array}{rrrr} 
1 & 0 & 0 & 0 \\ 
0 & 1 & 0 & 0 \\
0 & 0 & 1 & 0 \\
0 & 0 & 0 & 1
\end{array} \right ] \hspace{2mm} E = \displaystyle \left [ \begin{array}{r} 
1\\ 
1\\
1\\
1 \end{array} \right ]
\end{aligned}
\end{equation}

where $m_b$, represents the car body and the mass, $m_w$, represents the wheel assembly. The spring ,$k_s$, and damper, $b_s$, represent the passive spring and shock absorber placed between the car body and the wheel assembly.  The spring, $k_t$, models the compressibility of the pneumatic tire. The variables $x_b$, $x_w$ and $r$ are the car body travel, the wheel travel, and the road disturbance, respectively.  The force, $f_s$, which is applied between the body and the wheel assembly, is controlled by feedback. We assume that the disturbance force may enter at any of the state variables, not just the road disturbance. In MATLAB, we assign the following values to the parameters:
\begin{verbatim}
%----System Physical Constants----%
mb = 300;    % kg
mw = 60;     % kg
bs = 1000;   % N/m/s
ks = 16000 ; % N/m
kt = 190000; % N/m
\end{verbatim}

Therefore, the system matrices become:

\begin{verbatim}
% --------System Matrices(It is a SIMO system)---------%
A = [ 0 1 0 0; [-ks -bs ks bs]/mb ; ...
      0 0 0 1; [ks bs -ks-kt -bs]/mw];
B = [0 0; 0 10000/mb ; 0 0; [kt -10000]/mw];
C = eye(4);
E=[1;1;1;1];
\end{verbatim}
 
\begin{figure}[H]
    \centering
    \includegraphics[scale=.9]{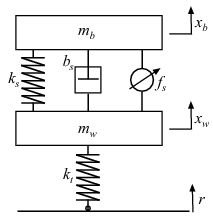}
    \caption{Free Body Diagram of the Quarter Car Suspension model (idealization).}
    \label{fig:car_sus_fbd}
\end{figure}

\color{lightgray} \begin{verbatim}
A =                                             B = 

   1.0e+03 *                                        1.0e+03 *

         0    0.0010         0         0                  0         0
   -0.0533   -0.0033    0.0533    0.0033                  0         0.0333
         0         0         0    0.0010                  0         0
    0.2667    0.0167   -3.4333   -0.0167                  3.1667   -0.1667

C =                                             E = 

     1     0     0     0                              1
     0     1     0     0                              1
     0     0     1     0                              1
     0     0     0     1                              1

\end{verbatim} \color{black}
Now we can obtain an optimal gain for control of the suspension system by using LQR below.  We will not dwell too much on this part of the design work since our main focus is the UIO design.  

\begin{verbatim}
%---------LQR design----------%
Q=[.25 0 0 0;0 4 0 0;0 0 1 0;0 0 0 4];
R=50*eye(2);
[K1 l s]=lqr(A,B,Q,R);
\end{verbatim}

In order to design the UIO, we first check to see if the $rank(CE) = rank(E) = 1$
\vspace{1em}
\begin{verbatim}
rank(C*E)
rank(E)
\end{verbatim}

        \color{lightgray} \begin{verbatim}
ans =       ans =

     1              1
\end{verbatim} \color{black}
    
\begin{figure}[H]
    \centering
    \includegraphics[scale=.3]{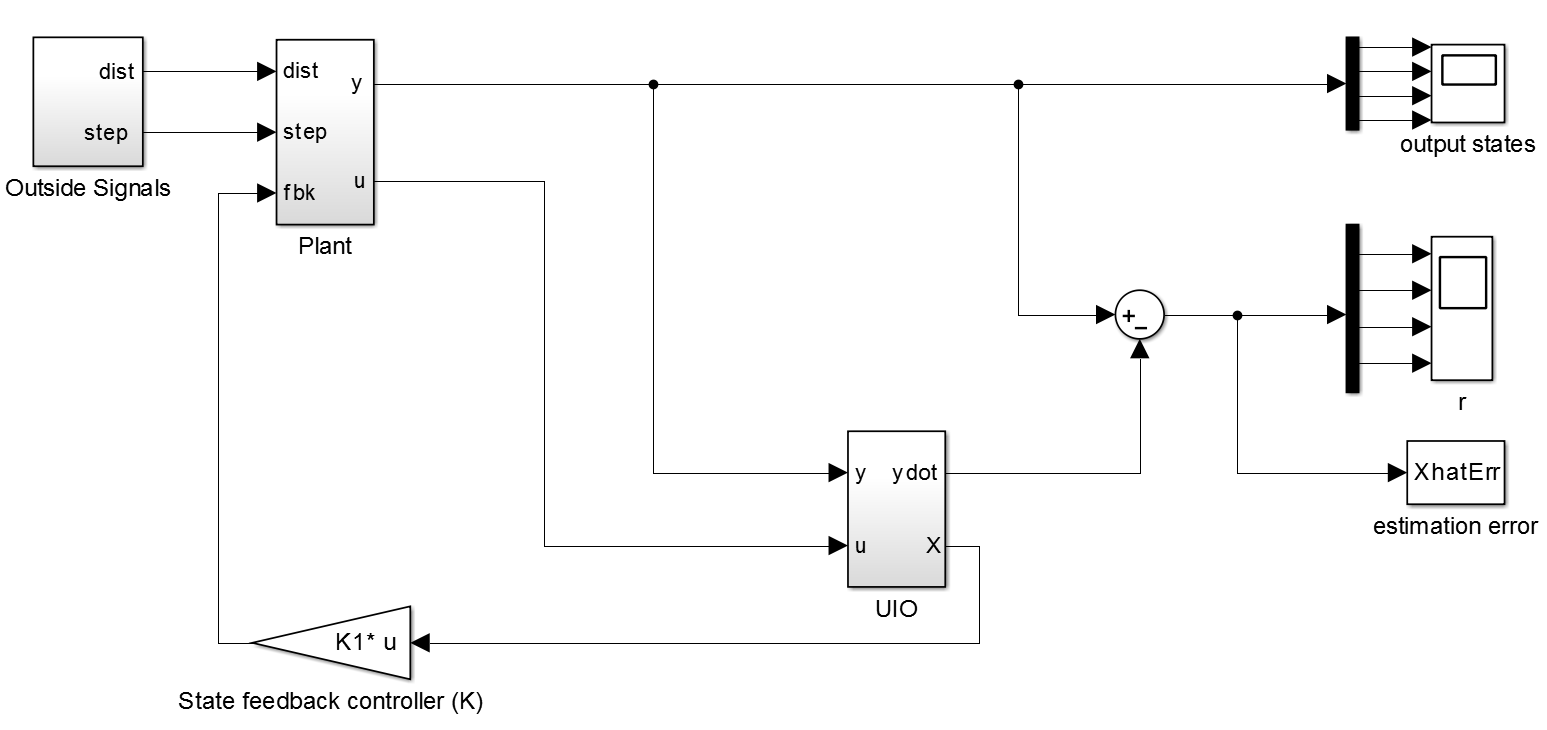}
    \caption{Root leve Simulink Block diagram architecture of the car suspension system in example two.}
    \label{fig:car_sus_fbd}
\end{figure}
Next we compute the observer matrices

\begin{verbatim}
H = E*inv((C*E)'*(C*E))*(C*E)'
T = eye(4)-H*C
A1 = T*A

f=zeros(4,4)
v=[-4 -4 -4 -4]
F=diag(v)

k1=inv(C)*(A-F-H*C*A)
k2=F*H
k=zeros
k=k1+k2
\end{verbatim}

\begin{figure}[H]
    \centering
    \includegraphics[scale=.4]{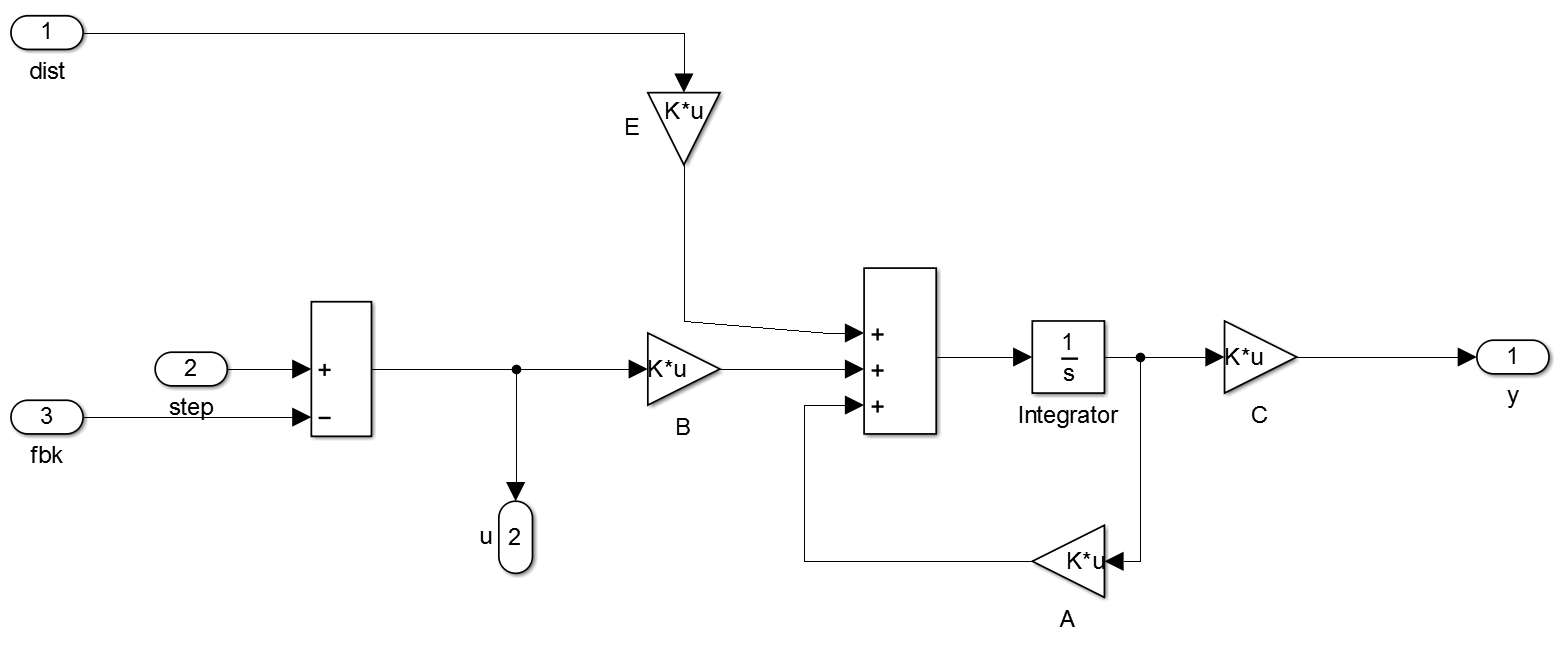}
    \caption{Quarter car suspension system dynamics implementation in Simulink.}
    \label{fig:carSusBD}
\end{figure}

\begin{figure}[H]
    \centering
    \includegraphics[scale=.45]{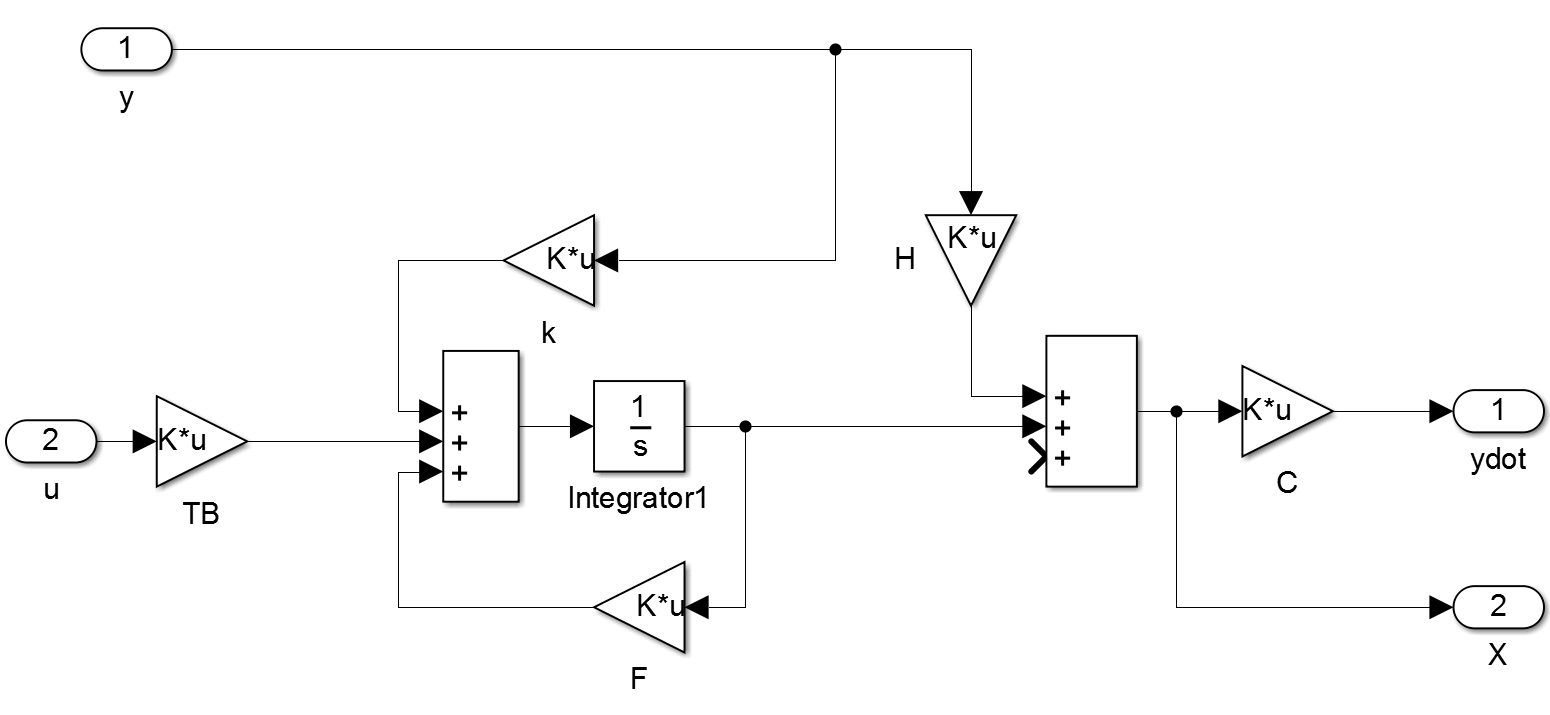}
    \caption{Simulink implementation of the Unknown Input Observer for the quarter car suspension model.}
    \label{fig:carSusUIO}
\end{figure}
        \color{lightgray} \begin{verbatim}
H =                                         T = 

    0.2500    0.2500    0.2500    0.2500         0.7500   -0.2500   -0.2500   -0.2500
    0.2500    0.2500    0.2500    0.2500        -0.2500    0.7500   -0.2500   -0.2500
    0.2500    0.2500    0.2500    0.2500        -0.2500   -0.2500    0.7500   -0.2500
    0.2500    0.2500    0.2500    0.2500        -0.2500   -0.2500   -0.2500    0.7500

A1 =                                            f = 

   1.0e+03 *

   -0.0533   -0.0026    0.8450    0.0031               0    0   0   0
   -0.1067   -0.0069    0.8983    0.0064               0    0   0   0
   -0.0533   -0.0036    0.8450    0.0041               0    0   0   0
    0.2133    0.0131   -2.5883   -0.0136               0    0   0   0

v =                                             F = eye(-4)

    -4    -4    -4    -4
    
    

k1 =                                            k2 = 4X4 matrix '-1'

   1.0e+03 *

   -0.0493   -0.0026    0.8450    0.0031
   -0.1067   -0.0029    0.8983    0.0064
   -0.0533   -0.0036    0.8490    0.0041
    0.2133    0.0131   -2.5883   -0.0096






k =

   1.0e+03 *

   -0.0503   -0.0036    0.8440    0.0021
   -0.1077   -0.0039    0.8973    0.0054
   -0.0543   -0.0046    0.8480    0.0031
    0.2123    0.0121   -2.5893   -0.0106

\end{verbatim} \color{black}

Next we form $A1$ and check observability

\begin{verbatim}
rank(obsv(A1,C))
\end{verbatim}

        \color{lightgray} \begin{verbatim}
ans =

     4

\end{verbatim} \color{black}
    
Since the observability matrix has full rank, we can go on and simulate the system to obtain the dynamic estimator performance. The initial conditions are set to be: 

\begin{verbatim}
X_0 = [-1;10;3;5]
\end{verbatim}

        \color{lightgray} \begin{verbatim}
X_0 =

    -1
    10
     3
     5

\end{verbatim} \color{black}

\begin{figure}
    \centering
    \includegraphics[scale=.45]{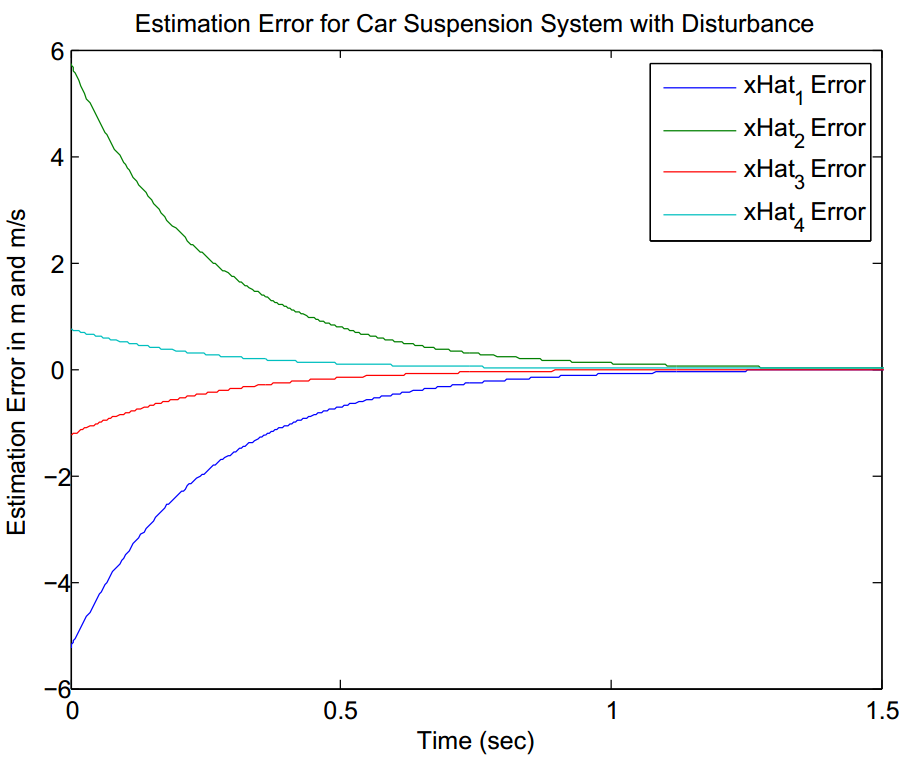}
    \caption{UIO results for the car suspension model.}
    \label{fig:carSusResults}
\end{figure}

Figure \ref{fig:carSusResults} shows the state estimation error for the UIO used in the car suspension model. The observer seems to converge within 1.5 seconds.  Figures \ref{fig:carSusUIO} shows the implementation of the UIO in Simulink.  All the files for this simulation can be found in the \color{lightgray}FDIR \color{black} folder. Specifically, the file \begin{verbatim} carsus.m \end{verbatim} and the Simulink model \begin{verbatim} car_sus.xsl \end{verbatim} contain everything. 

\subsection{Coupled Mass-Spring System}

\begin{figure}[H]
    \centering
    \includegraphics[scale=.55]{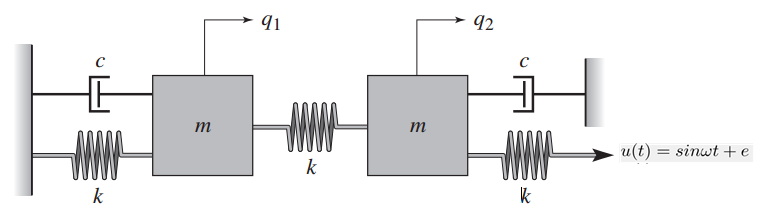}
    \caption{The free body diagram for the two mass and spring system.}
    \label{fig:dualMS}
\end{figure}

Consider the coupled spring mass system shown in Figure \ref{fig:dualMS}. The coupled mass spring system is a well studied mechanical system in control theory. The input to this system is the sinusoidal motion of the end of the rightmost spring, and the output is the position of each mass, $q_1$ and $q_2$. The disturbance input is also shown in the figure as $e$. In state space form, the equations of motion are: 
\begin{equation}
\begin{aligned}
x = \displaystyle \left [ \begin{array}{r}
q_1 \\
q_2 \\
\dot{q}_1 \\
\dot{q}_2 \\
\end{array} \right ], \hspace{2mm}
A = \displaystyle \left [ \begin{array}{rrrr}
0 & 0 & 1 & 0 \\
0 & 0 & 0 & 1 \\
\frac{-2k}{m} & \frac{k}{m} & \frac{-c}{m} & 0 \\
\frac{k}{m} & \frac{-2k}{m} & 0 & \frac{-c}{m}
 \end{array} \right ], \hspace{2mm} 
B= \displaystyle \left [ \begin{array}{c} 
 0 \\
 0\\
 0  \\
\frac{k}{m}
 \end{array} \right ], \hspace{2mm} 
C = \displaystyle \left [ \begin{array}{rrrr} 
1 & 0 & 0 & 0 \\ 
0 & 1 & 0 & 0 \\
\end{array} \right ], \hspace{2mm} 
E = \displaystyle \left [ \begin{array}{r} 
0\\ 
0\\
0\\
1 \end{array} \right ]
\end{aligned}
\end{equation}
In MATLAB, we assign the following values to the parameters:
\begin{verbatim}
%----------System Parameters---------%
k   = 1   %
m   = 1   % kg
c   = 1 %
\end{verbatim}

        \color{lightgray} \begin{verbatim}
k =         m =         c = 

     1          1               1

\end{verbatim} \color{black}
to obtain the following state space formulation: \\ 
    \begin{verbatim}
%----------State Space Formulation---------%
A   = [0 0 1 0;
       0 0 0 1;
       -2*k/m k/m -c/m 0;
       k/m -2*k/m 0 -c/m]

B   = [0;0;0;k/m]
C   = [1 0 0 0;0 1 0 0;0 0 1 0;0 0 0 1]
D   = 0
E   = [0;0;0;1]
\end{verbatim}

        \color{lightgray} \begin{verbatim}
A =                         B =             C =                  D =        E = 

     0  0   1   0               0                1  0   0   0        0              0
     0  0   0   1               0                0  1   0   0                       0
    -2  1  -1   0               0                0  0   1   0                       0
     1 -2   0  -1               1                0  0   0   1                       1

\end{verbatim} \color{black}
In order to control the placement of the masses in the system, we utilized LQR.  The following MATLAB code accomplishes this objective: 

\begin{figure}[H]
    \centering
    \includegraphics[scale=.35]{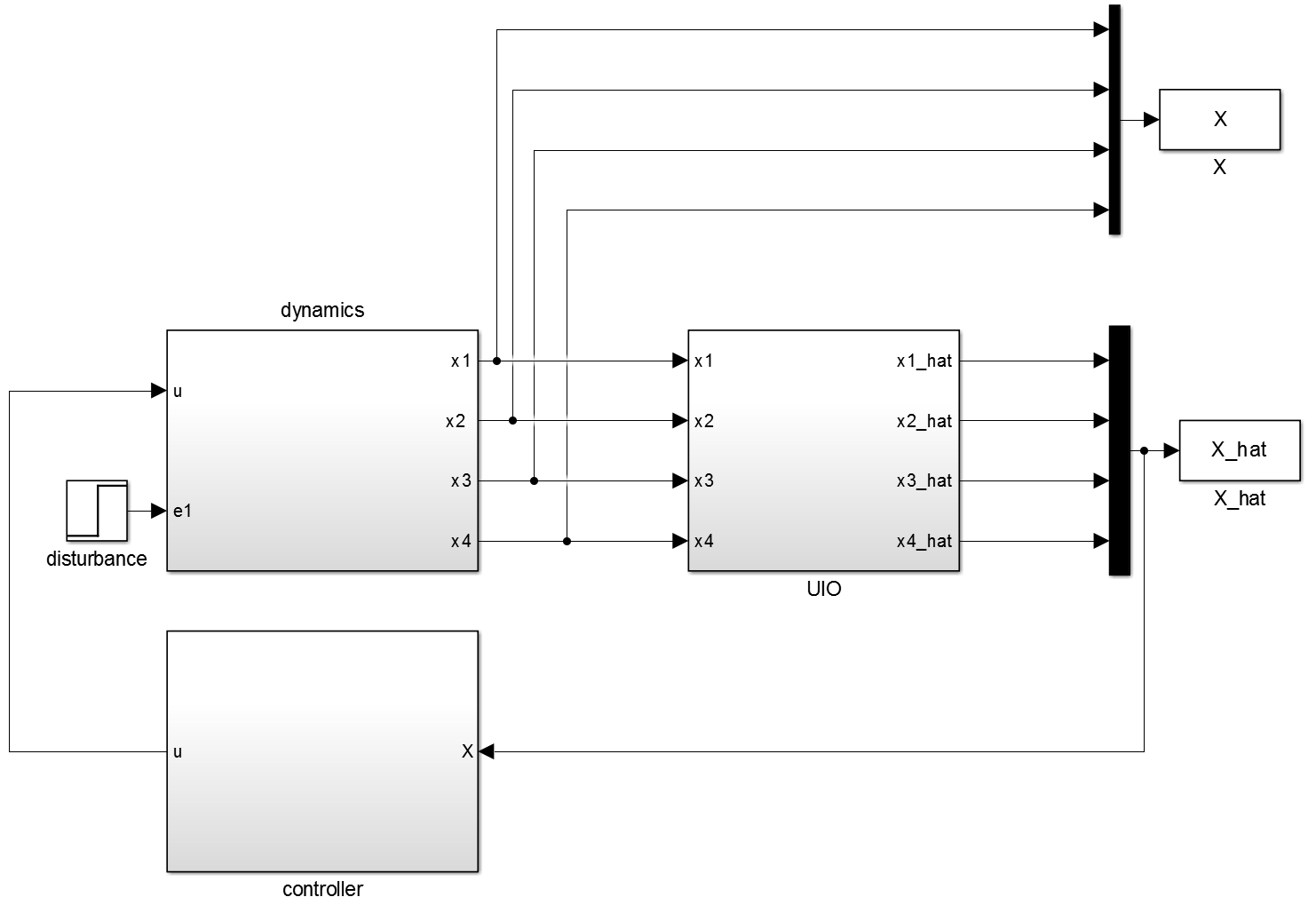}
    \caption{The root level Simulink block diagram for the coupled mass spring dynamic system.}
    \label{ex:msMDLroot}
\end{figure}

    \begin{verbatim}
%----------LQR Control---------%
Q=[100 0 0 0;
    0 100 0 0;
    0 0 100 0;
    0 0 0 100]
R=1
[Klqr ll ss]=lqr(A,B,Q,R)
\end{verbatim}

        \color{lightgray} \begin{verbatim}
Q =                     R =       Klqr = 

   100  0    0    0         1           -1.3620   10.4615    2.0165   10.0419
     0  100  0    0
     0  0    100  0
     0  0    0    100

\end{verbatim} \color{black}
Now we are ready to begin the UIO design.  The first step is to check that $rank(CE) = rank(E)$:  
    \begin{verbatim}
%----------Step 1: rank(CE) = rank(E) ?= 1---------%
rank(C*E)
rank(E)
\end{verbatim}

        \color{lightgray} \begin{verbatim}
ans =     ans = 

     1          1
\end{verbatim} \color{black}
Next we are ready to compute the first set of observer matrices:
    \begin{verbatim}
%----------Step 2: Compute observer matrices---------%
H = E*inv((C*E)'*(C*E))*(C*E)'
T = eye(4)-H*C
A1 = T*A
\end{verbatim}

        \color{lightgray} \begin{verbatim}
H =                 T =                     A1 = 

     0  0   0   0       1   0   0   0               0   0   1   0
     0  0   0   0       0   1   0   0               0   0   0   1
     0  0   0   0       0   0   1   0              -2   1  -1   0
     0  0   0   1       0   0   0   0               0   0   0   0

\end{verbatim} \color{black}
\begin{figure}[H]
    \centering
    \includegraphics[scale=.45]{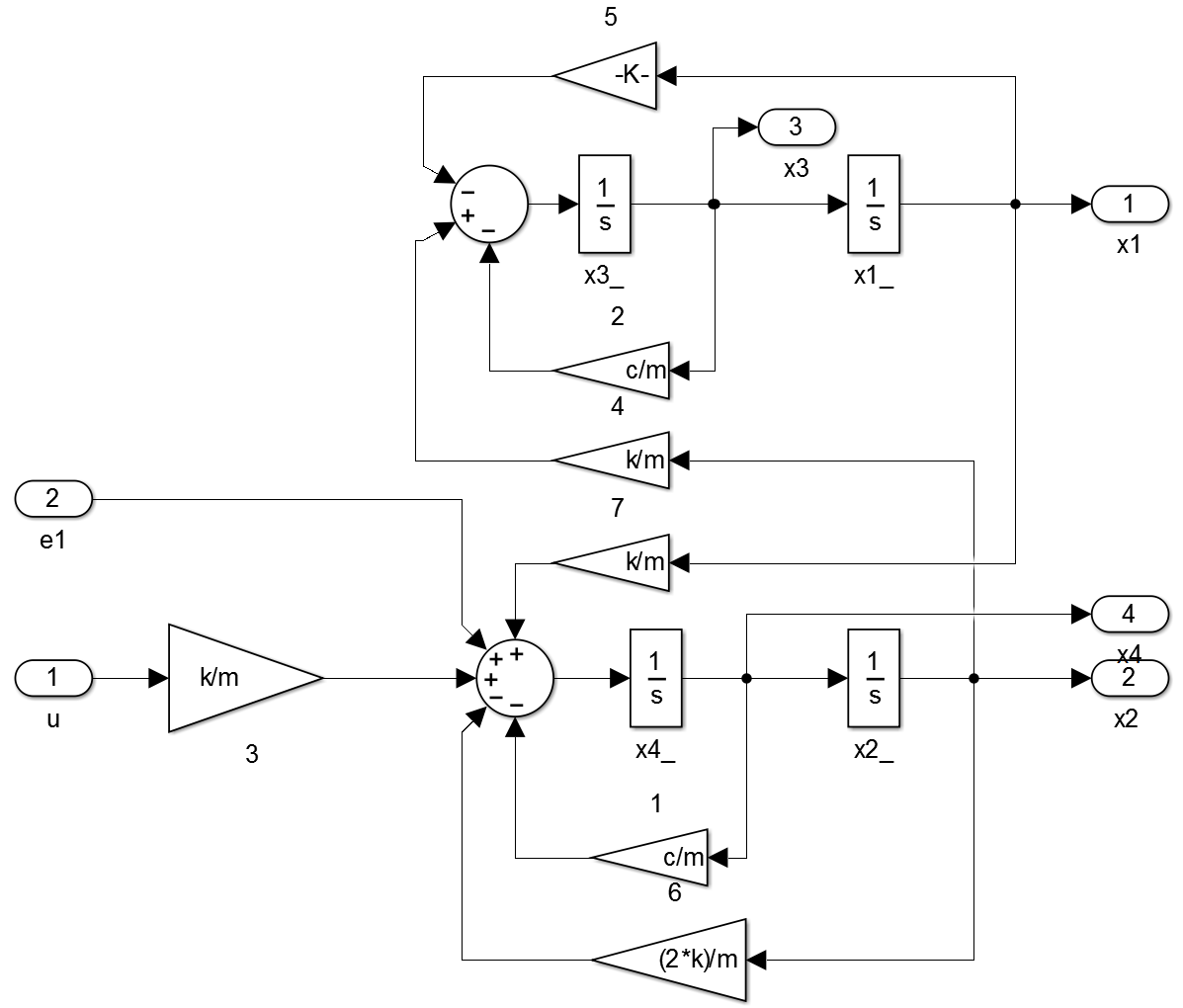}
    \caption{The mass spring system dynamics implemented in linearized form in Simulink.}
    \label{fig:msDynSim}
\end{figure}

Now that $A_1$ has been obtained, we must check the observability of the pair $(C,A_1)$:
    \begin{verbatim}
%----------Step 3:Check (C,A1) rank---------%
rank(obsv(A1,C)) 
\end{verbatim}
        \color{lightgray} \begin{verbatim}
ans =

     4

\end{verbatim} \color{black}
Since the pair $(C,A_1)$ is full rank, it is observable.  Therefore, we can use pole placement for the observer gain $K_1$:
    \begin{verbatim}
%----------Step 4: Choose observer poles---------%
K1 = place(A',C',[-2,-10,-5,-3])'
\end{verbatim}

        \color{lightgray} \begin{verbatim}
K1 =

     2     0     1     0
     0    10     0     1
    -2     1     4     0
     1    -2     0     2

\end{verbatim} \color{black}

\begin{figure}[H]
    \centering
    \includegraphics[scale=.5]{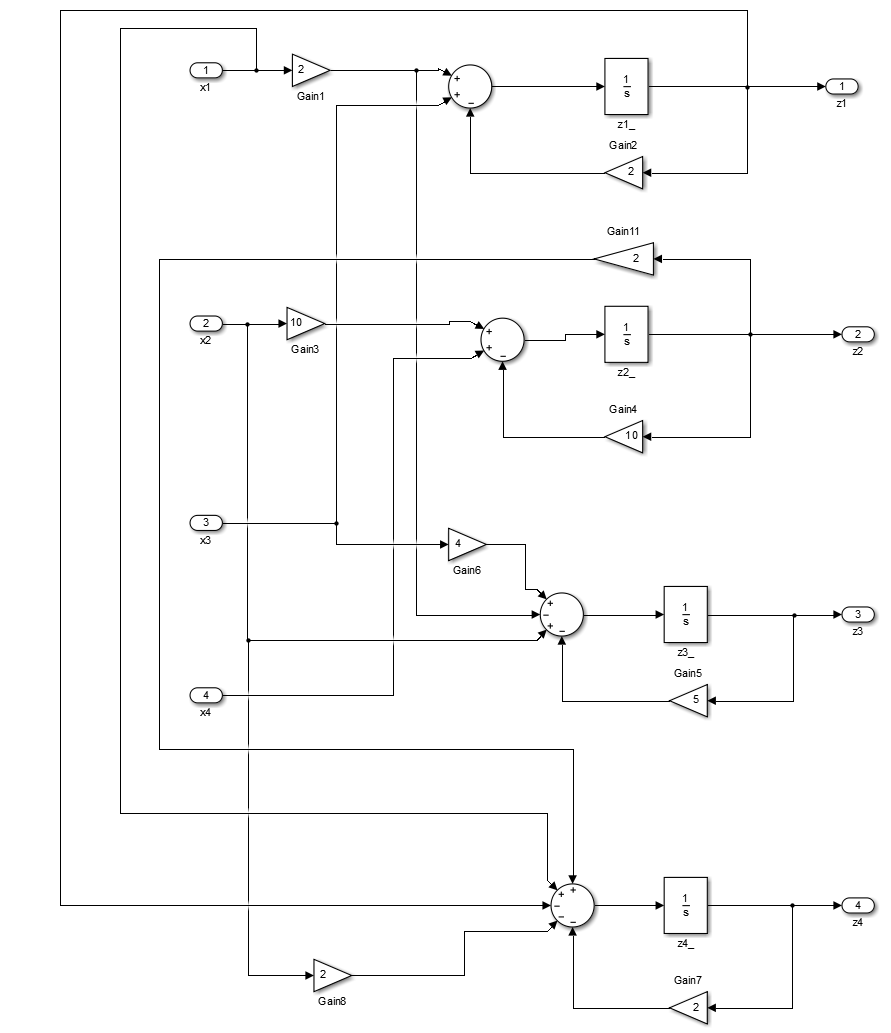}
    \caption{The mass spring system dynamics implemented in linearized form in Simulink.}
    \label{fig:msUIO}
\end{figure}

and to conclude the observer design we compute the $F$ and $K$ matrices:
    \begin{verbatim}
%----------Step 5: Finish observer design---------%
F = A1-K1*C
K = K1 + F*H
\end{verbatim}

        \color{lightgray} \begin{verbatim}
F =                     K = 

    -2  0   0   0           2   0   1   0
     0 -10  0   0           0   10  0   1
     0  0  -5   0          -2   1   4   0
    -1  2   0  -2           1  -2   0   0 
\end{verbatim} \color{black}
The simulation results are considered next.  The top level Simulink block diagram is shown in Figure \ref{ex:msMDLroot}.  The mass spring system dynamics are implemented as shown in Figure \ref{fig:msDynSim}. The UIO implementation is shown in Figure \ref{fig:msUIO}. The state estimation error is shown in Figure \ref{fig:UIOerr}. The following initial conditions were used for the simulation. 

\begin{verbatim}
x1_0       = -1;
x2_0       = -5;
x3_0       = 1;
x4_0       = 1; 
\end{verbatim}

\begin{figure}[H]
    \centering
    \includegraphics[scale=.5]{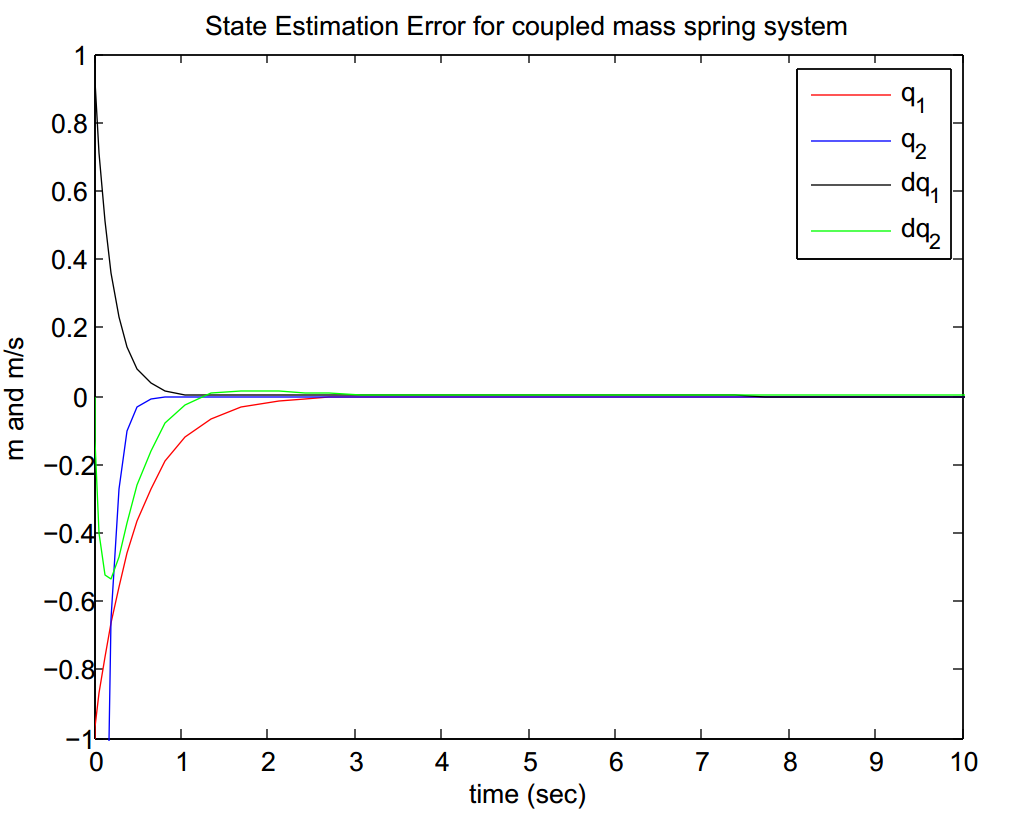}
    \caption{The state estimation error for the mass spring system using the UIO.}
    \label{fig:UIOerr}
\end{figure}
\section{Conclusion}
This memo begins by motivating the UIO paradigm for state estimation.  A main advantage of this approach is to decouple the plant disturbance inputs from the state estimation process.  The UIO design procedure was reviewed and illustrated with three common examples found in control theory.  The estimation error for the coupled mass spring system, the car suspension system and the linear system with no inputs was considered. In each example, it was shown that this error asymptotically approaches zero by utilization of a UIO.
\bibliography{report}
\bibliographystyle{amsplain}

\end{document}